\documentclass[a4paper]{article}

\usepackage{amsfonts,amsmath}
\usepackage{amstext,amssymb}
\usepackage{amsthm}
\usepackage{tikz}
\usepackage{xspace}
\usepackage{boxedminipage}
\usepackage[linesnumbered, boxed, algosection,noline]{algorithm2e}

\usepackage{thm-restate}
\usepackage{thmtools}
\usepackage{enumitem}

\usepackage[textwidth=.9in]{todonotes}
\usepackage[margin=1.2in]{geometry}
\usepackage{authblk}
\usepackage{color}

\newcommand{\Yes}{\textsc{Yes}\xspace}
\newcommand{\No}{\textsc{No}\xspace}

\newcommand{\phisub}{\textsc{$\Phi$-Sub\-set}\xspace}

\newcommand{\phiext}{\textsc{$\Phi$-Ex\-ten\-sion}\xspace}

\newcommand{\fvsext}{\textsc{Feed\-back Ver\-tex Set Ex\-ten\-sion}\xspace}

\newcommand{\bitsize}{N}
\newcommand{\cF}{\mathcal{F}}

\newcommand{\cM}{\mathcal{M}}

\newcommand{\mif}{\texttt{mif}\xspace}

\newcommand{\rtfvs}{$O(1.7117^n)$\xspace}
\newcommand{\rtenumfvs}{$O(1.8527^n)$\xspace}
\newcommand{\rthittingset}{$O(1.6627^n)$\xspace}

\newcommand{\bs}{\backslash}

\newcommand{\wA}{\alpha_1}
\newcommand{\wB}{\alpha_2}
\newcommand{\wBeta}{w_{\beta}}
\newcommand{\wK}{w_k}
\newcommand{\sergeBeta}{\beta}
\newcommand{\sergeGamma}{\gamma}

\newtheorem{theorem}{Theorem}
\newtheorem{lemma}[theorem]{Lemma}
\newtheorem{definition}[theorem]{Definition}
\newtheorem{proposition}{Proposition}[section]

\newcommand{\pbDef}[3]{%
    \noindent
    \begin{center}
        \begin{boxedminipage}{0.99 \textwidth}
            {#1}
            \smallskip\\
            \begin{tabular}{lp{0.99 \textwidth - \widthof{~~~Question: }}}
                Input:&#2\\
                Question:&#3
            \end{tabular}
        \end{boxedminipage}
    \end{center}
}

\newcommand{\pbDefOpt}[3]{%
    \noindent
    \begin{center}
        \begin{boxedminipage}{0.99 \textwidth}
            {#1}
            \smallskip\\
            \begin{tabular}{lp{0.99 \textwidth - \widthof{~~~Output: }}}
                Input:&#2\\
                Output:&#3
            \end{tabular}
        \end{boxedminipage}
    \end{center}
}

\date{}

\title{Exact Algorithms via Multivariate Subroutines}

\author[1,2]{Serge Gaspers}
\author[1,2]{Edward J. Lee}
\affil[1]{UNSW Australia, Sydney, Australia.\\ \texttt{sergeg@cse.unsw.edu.au}, \texttt{e.lee@unsw.edu.au}}
\affil[2]{Data61, CSIRO, Sydney, Australia}

\begin{document}
    \maketitle

    \begin{abstract}
        We consider the family of \phisub problems, where the input consists of an instance $I$ of size $N$ over a universe $U_I$ of size $n$ and the task is to check whether the universe contains a subset with property $\Phi$ (e.g., $\Phi$ could be the property of being a feedback vertex set for the input graph of size at most $k$).
        Our main tool is a simple randomized algorithm which solves \phisub in time $(1+b-\frac{1}{c})^n N^{O(1)}$, provided that there is an algorithm for the \phiext problem with running time $b^{n-|X|} c^k N^{O(1)}$. Here, the input for \phiext is an instance $I$ of size $N$ over a universe $U_I$ of size $n$, a subset $X\subseteq U_I$, and an integer $k$, and the task is to check whether there is a set $Y$ with $X\subseteq Y \subseteq U_I$ and $|Y\setminus X|\le k$ with property $\Phi$.
        We also derandomize this algorithm at the cost of increasing the running time by a subexponential factor in $n$, and we adapt it to the enumeration setting where we need to enumerate all subsets of the universe with property $\Phi$.
        This generalizes the results of Fomin et al. [STOC 2016] who proved them for the case $b=1$.
        As case studies, we use these results to design faster deterministic algorithms for
        \begin{itemize}
            \item checking whether a graph has a feedback vertex set of size at most $k$,
            \item enumerating all minimal feedback vertex sets,
            \item enumerating all minimal vertex covers of size at most $k$, and
            \item enumerating all minimal 3-hitting sets.
        \end{itemize}
        We obtain these results by deriving new $b^{n-|X|} c^k N^{O(1)}$-time algorithms for the corresponding \phiext problems (or the enumeration variant). In some cases, this is done by simply adapting the analysis of an existing algorithm, in other cases it is done by designing a new algorithm. Our analyses are based on Measure and Conquer, but the value to minimize, $1+b-\frac{1}{c}$, is unconventional and leads to non-convex optimization problems in the analysis.
    \end{abstract}
    
\section{Introduction}

    In exponential-time algorithmics \cite{FominK10}, the aim is to design algorithms for NP-hard problems with the natural objective to minimize their running times.
    In this paper, we consider a broad class of subset problems, where for an input instance $I$ on a universe $U_I$, the question is whether there is a subset $S$ of the universe satisfying certain properties.
    For example, in the \textsc{Feedback Vertex Set} problem, the input instance consists of a graph $G=(V,E)$ and an integer $k$, the universe is the vertex set and the property to be satisfied by a subset $S$ is the conjunction of ``$|S|\le k$'' and ``$G-S$ is acyclic''.

    More formally, and using definitions from \cite{FominGLS16}, an \emph{implicit set system} is a function $\Phi$ that takes as input a string $I \in \{0,1\}^*$ and outputs a set system $(U_I, \cF_I)$, where $U_I$ is a universe  and $\cF_I$ is a collection of subsets of $U_I$.  The string $I$ is referred to as an \emph{instance}  and we denote by $|U_I| = n$ the size of the universe and by $|I|=\bitsize$ the size of the instance. 
    We assume that $\bitsize\ge n$. The implicit set system $\Phi$ is \emph{polynomial time computable} if (a) there exists a polynomial time algorithm that given $I$ produces $U_I$, and (b) there exists a polynomial time algorithm that given $I$, $U_I$ and a subset $S$ of $U_I$ determines whether $S \in \cF_I$.
    All implicit set systems discussed in this paper are polynomial time computable.

    \pbDefOpt{\phisub}{An instance $I$}{A set $S \in {\cal F}_I$ if one exists.}
    \pbDef{\phiext}{An instance $I$, a set $X \subseteq U_I$, and an integer $k$.}
    {Does there exists a subset $S \subseteq (U_I \setminus X)$ such that $S \cup X \in {\cal F}_I$ and $|S| \leq k$?}

    In recent work, Fomin et al. \cite{FominGLS16} showed that $c^k N^{O(1)}$ time algorithms ($c\in O(1)$) for \phiext lead to competitive exponential-time algorithms for many \phisub problems. The main tool was a simple randomized algorithm which solves \phisub in time $(2-\frac{1}{c})^n N^{O(1)}$ if there is an algorithm that solves \phiext in time $c^k N^{O(1)}$. A derandomization was also given, turning the randomized algorithm into a deterministic one at the cost of a $2^{o(n)}$ factor in the running time. The method was also adapted to enumeration algorithms and combinatorial upper bounds.
    This framework, together with a large body of work in parameterized algorithmics \cite{downey2013fundamentals}, where $c^k N^{O(1)}$ time algorithms are readily available for many subset problems, led to faster algorithms for around 30 decision and enumeration problems.

    In this paper, we extend the results of Fomin et al. \cite{FominGLS16} and show that a $b^{n-|X|} c^k N^{O(1)}$ time algorithms ($b,c\in O(1)$) for \phiext lead to randomized $(1+b-\frac{1}{c})^n N^{O(1)}$ time algorithms for \phisub. Our result can be similarly derandomized and adapted to the enumeration setting. Observe that for $b=1$, the results of \cite{FominGLS16} coincide with ours, but that ours have the potential to be more broadly applicable and to lead to faster running times. The main point is that if we use a $c^k N^{O(1)}$ time algorithm as a subroutine to design an algorithm exponential in $n$, we might as well allow a small exponential factor in $n$ in the running time of the subroutine.

    Similar as in \cite{FominGLS16}, the \phiext problem can often be solved by preprocessing the elements in $X$ in a simple way and then using an algorithm for a subset problem. In the case of \textsc{Feedback Vertex Set}, the vertices in $X$ can simply be deleted from the input graph. Whereas the literature is rich with $c^k N^{O(1)}$ time algorithms for subset problems, algorithms with running times of the form $b^{n} c^k N^{O(1)}$ with $b>1$ are much less common.%
    \footnote{One notable exception is by Eppstein \cite{eppstein2003small}, who showed that all maximal independent sets of size at most $k$ in a graph on $n$ vertices can be enumerated in time $(4/3)^n (81/64)^k n^{O(1)}$.}
    One issue is that there is, in general, no obviously best trade-off between the values of $b$ and $c$ for such algorithms. However, the present framework gives us a precise objective: we should aim for values of $b$ and $c$ that minimize the base of the exponent, $(1+b-\frac{1}{c})$.

    Our applications consist of three case studies centered around some of the most fundamental problems considered in \cite{FominGLS16}, feedback vertex sets and hitting sets. For the first case study, we considered the \textsc{Feedback Vertex Set} problem: given a graph $G$ and an integer $k$, does $G$ have a feedback vertex set of size at most $k$?
    For this problem, we re-analyze the running time of the algorithm from \cite{fomin2008minimum}. In \cite{fomin2008minimum,gaspers2010exponential}, the algorithm was analyzed using Measure and Conquer: using a measure that is upper bounded by $\alpha n$ and aiming for a running time of $2^{\alpha n} n^{O(1)}$ the analysis of the branching cases led to constraints lower bounding the measure and the objective was to minimize $\alpha$ subject to these constraints. In our new analysis, we add an additive term $w_k\cdot k$ to the measure and adapt the constraints accordingly. If all constraints are satisfied, we obtain a running time of $2^{\alpha n+w_k k} n^{O(1)}$. Our framework naturally leads us to minimize $2^\alpha-2^{-w_k}$. This approach leads to a $O(1.5422^n \cdot 1.2041^k)$ time algorithm, which, combined with our framework gives a deterministic \rtfvs time algorithm for \textsc{Feedback Vertex Set}. This improves on previous results giving $O(1.8899^n)$ \cite{razgon2006exact}, $O(1.7548^n)$ \cite{fomin2008minimum}, $O(1.7356^n)$ \cite{xiao2013exact}, $O(1.7347^n)$ \cite{fomin2015large}, and $O(1.7216^n)$ \cite{FominGLS16} time algorithms for the problem. We note that adapting the analysis of other existing exact and parameterized algorithms did not give faster running times. Also, if we allow randomization, the $O(1.6667^n)$ time algorithm by \cite{FominGLS16} (which can also be achieved using our framework) remains fastest.

    Our second case study is more involved. Simply using an existing algorithm and adapting the measure was not sufficient to improve upon the best known enumeration algorithms (and combinatorial upper bounds) for minimal feedback vertex sets.
    Here, the task is, given a graph $G$, to output each feedback vertex set that is not contained in any other feedback vertex set.
    We design a new algorithm for enumerating all minimal feedback vertex sets. We also need a new combinatorial upper bound for the number of minimal vertex covers of size at most $k$ to handle one special case in the enumeration of minimal feedback vertex sets.%
    \footnote{Previous work \cite{Byskov04,eppstein2003small} along these lines focused on small maximal independent sets in the context of graph coloring, whose bounds were insufficient for our purposes. Here, we need better bounds on large maximal independent sets or small minimal vertex covers.}
    We obtain a $O(1.7183^n \cdot 1.1552^k)$ time algorithm for enumerating all minimal feedback vertex sets. Our framework thus leads to a running time of \rtenumfvs, improving on the previous best bound of $O(1.8638^n)$ \cite{fomin2008minimum}.
    The current best lower bound for the number of minimal feedback vertex sets is $O(1.5926^n)$ \cite{fomin2008minimum}.
    We would like to highlight that the enumeration of minimal feedback vertex sets is completely out of scope for the more restricted framework of \cite{FominGLS16}: the number of minimal feedback vertex sets of size at most $k$ cannot be upper bounded by $c^k n^{O(1)}$, as evidenced by a disjoint union of $k$ cycles of length $n/k$.

    Our last case study gives a new algorithm for enumerating all minimal 3-hitting sets, also known as minimal transversals of rank-3 hypergraphs. These are minimal sets $S$ of vertices of a hypergraph where each hyperedge has size at most 3 such that every hyperedge contains at least one vertex of $S$. We re-analyze an existing algorithm \cite{cochefert2016faster} for this enumeration problem, adapting the measure in a similar way as in the first case study, and we obtain a multivariate running time of $O(1.5135^n \cdot 1.1754^k)$, leading to an \rthittingset time enumeration algorithm. This breaks the natural time bound of $O(1.6667^n)$ of the previously fastest algorithm \cite{FominGLS16}. The current best lower bound gives an infinite family of rank-3 hypergraphs with $\Omega(1.5848^n)$ minimal transversals \cite{cochefert2016faster}.

\section{Preliminaries}
    Let $G = (V,E)$ be a graph with a set of vertices $V$ and a set of edges 
    $E \subseteq \{uv : u, v \in V\}$.
    The \emph{degree} $d(u)$ of a vertex $u$ is the number of neighbors of $u$ in $G$. The \emph{degree} of a graph $\Delta(G)$ is the maximum $d(u)$ across all $u \in V$. A graph $G' = (V', E')$ is a \emph{subgraph} of $G$ if $V' \subseteq V$ and $E' \subseteq E$ and $G'$ is an \emph{induced subgraph} of $G$ if, in addition, $G$ has no edge $uv$ with $u,v\in V'$ but $uv\notin E'$. In this case, we also denote $G'$ by $G[V']$.
    A forest is an acyclic graph.
    A subset $F\subseteq V$ is acyclic if $G[F]$ is a forest.
    An acyclic subset $F\subseteq V$ is \emph{maximal} in $G$ if it is not a subset of any other acyclic subset.
    For an acyclic subset $F\subseteq V$, we denote the set of maximal acyclic supersets of $F$ as $\cM_G(F)$ and the set of maximum (i.e., largest) acyclic supersets of $F$ as $\cM^*_G(F)$.

    Let $T$ be a subgraph of $G$. Define Id$(T,t)$ as an operation on $G$ which contracts all edges of $T$ into one vertex $t$, removing induced loops. This may create multiedges in $G$. Define Id$^*(T,t)$ as the operation Id$(T,t)$ followed by removing all vertices connected to $t$ by multiedges.
    A \emph{non-trivial} component of a graph $G$ is a connected component on at least two vertices.
    The following propositions from \cite{fomin2008minimum} will be useful.
    \begin{proposition}{\textsc{\cite{fomin2008minimum}}}
        \label{prop:1}
        Let $G = (V, E)$ be a graph, $F \subseteq V$ be an acyclic subset of vertices and $T$ be a non-trivial component of $G[F]$. Denote by $G'$ the graph obtained from $G$ by the operation Id$^*(T,t)$ and let $F' = F \cup \{t\} \bs T$. Then for $X' = X \cup \{t\} \bs T$ where $X, X' \subseteq V$
        \begin{itemize}
             \item $X \in \cM_G(F)$ if and only if $X' \in \cM_{G'}(F')$, and
             \item $X \in \cM^*_{G}(F)$ if and only if $X' \in \cM^*_{G'}(F')$.
         \end{itemize} 
    \end{proposition}

    \noindent
    Using operation Id$^*$ on each non-trivial component of $G[F]$, results in an independent set $F'$.


    \begin{proposition}{\textsc{\cite{fomin2008minimum}}}
        \label{prop:3}
        Let $G = (V , E)$ be a graph and $F$ be an independent set in $G$ such that $V \bs F = N(t)$ for some $t \in F$. Consider the graph $G' = G[N(t)]$ and for every pair of vertices $u,v \in  N(t)$ having a common neighbor in $F \bs \{t\}$ add an edge $uv$ to $G'$. Denote the obtained graph by $H$ and let $I \subseteq N(t)$. Then $F \cup I \in \cM_G(F)$ if and only if $I$ is a maximal independent set in $H$ . In particular, $F \cup I \in \cM^*_G(F)$ if and only if $I$ is a maximum independent set in $H$.
    \end{proposition}

    For an acyclic subset $F$, a so-called \emph{active} vertex $t\in F$ and a neighbor $v\in N(t)\setminus F$, we will now define the concept of generalized neighbors of $v$, as well as their generalized neighbors.
    Denote by $K$ the set of vertices of $F$ adjacent to $v$ other than $t$. Let $G'$ be the graph obtained after the operation Id$(K \cup \{v\}, u)$. A vertex $w \in V(G') \backslash \{t\}$ is a \textit{generalized neighbor} of $v$ in $G$ if $w$ is a neighbor of $u$ in $G'$. Denote by gd$(v)$ the \textit{generalized degree} of $v$ which is its number of generalized neighbors.
    For a given generalized neighbor $x$ of $v$, denote by $K'$ the set of vertices in $F$ adjacent to $x$. Denote $G''$ as the graph obtained after the operation Id$(K' \cup \{x\}, u')$. A \emph{generalized neighbor} of $x$ is any vertex $y \in V(G'') \bs \{v\}$ which is adjacent to $u'$ in $G''$. Also use the notation gd$(x)$ to represent the \emph{generalized degree} of $x$, which is a very similar notion to that of gd$(v)$.
    Lastly, all randomized algorithms in this paper are Monte Carlo algorithms with one-sided error. On \No-instances they always return \No, and on \Yes-instances they return \Yes (or output a certificate) with probability $>\frac{1}{2}$.

\section{Results}
    \noindent
    Our first main result gives exponential-time randomized algorithms for \phisub based on single-exponential multivariate algorithms for \phiext with parameter $k$.

    \begin{restatable}{theorem}{thmMainOne}
        \label{thm:main1}
        If there is an algorithm for \phiext with running time $b^{n - |X|} c^k  \bitsize^{O(1)}$ then there is a randomized algorithm for \phisub with running time $(1+b-\frac{1}{c})^{n}\bitsize^{O(1)}$.
    \end{restatable}

    \noindent
    The next main result derandomizes the algorithm of Theorem \ref{thm:main1} at a cost of a subexponential factor in $n$ in the running time.

    \begin{restatable}{theorem}{thmMainTwo}
        \label{thm:main2}
        If there is an algorithm for \phiext with running time $b^{n - |X|} c^k  \bitsize^{O(1)}$ then there is an algorithm for \phisub with running time $(1+b-\frac{1}{c})^{n + o(n)}\bitsize^{O(1)}$.
    \end{restatable}

    \noindent
    We require the following notion of \textit{(b,c)-uniform} to describe our enumeration algorithms. Let $c,b \geq 1$ be real valued constants and $\Phi$ be an implicit set system. Then $\Phi$ is \textit{(b,c)}-uniform if for every instance $I$, set $X \subseteq U_I$, and integer $k \leq n - |X|$, the cardinality of the collection
    $
        \mathcal{F}_{I,X}^k = \{ S \subseteq U_I \bs X : |S| = k \text{ and } S \cup X \in \mathcal{F}_I \}
    $
    is at most $b^{n - |X|} c^k n^{O(1)}$. Then the following theorem provides new combinatorial bounds for collections generated by $(b,c)$-uniform implicit set systems.

    \begin{restatable}{theorem}{thmMainThree}
        \label{thm:main3}
        Let $c,b \geq 1$ and $\Phi$ be an implicit set system. If $\Phi$ is $(b,c)$-uniform then $|\mathcal{F}_I| \leq \left(  1 + b - \frac{1}{c} \right)^n n^{O(1)}$ for every instance $I$.
    \end{restatable}

    \noindent
    We say that an implicit set system is \textit{efficiently $(b,c)$-uniform} if there exists an algorithm that given $I, X$ and $k$ enumerates all elements of $\cF_{X,I}^k$ in time $b^{n - |X|}c^k \bitsize^{O(1)}$. In this case, we enumerate $\cF_I$ in the same time, up to a subexponential factor in $n$.

    \begin{restatable}{theorem}{thmMainFour}
        \label{thm:main4}
        Let $c,b \geq 1$ and $\Phi$ be an implicit set system. If $\Phi$ is efficiently $(b,c)$-uniform then there is an algorithm that given as input $I$ enumerates $\cF_I$ in time $\left( 1 + b - \frac{1}{c} \right)^{n+o(n)} \bitsize^{O(1)}$.
    \end{restatable}

\section{Random Sampling and Multivariate Subroutines}
    In this section, we prove Theorem \ref{thm:main1}.
    To do this, we first need the following lemmas.

    \begin{lemma}
        \label{thm:case1inequality}
        If $b, c \geq 1$ then

        \[ b \cdot  c^{\frac{1}{bc}} \leq 1 + b - \frac{1}{c} \]
    \end{lemma}

    \begin{proof}
        As both sides of the inequality are positive, it suffices to show that
        $ \log (b   c^{\frac{1}{bc}}) \leq \log(1 + b - 1/c) $.
        So we let $y = \log(1 + b - 1/c) - \log b - \frac{1}{bc} \log c$ and prove that $y \geq 0$ for all $b,c \geq 1$.
        When $c = 1$ we have that $y = 0$ for all $b$.
        We will show that for any fixed $b \geq 1$ we have that $y \geq 0$ by showing that $y$ increases with $c\ge 1$.
        For fixed $b$, the partial derivative with respect to $c$ is
        $ \frac{\partial y}{\partial c} = \frac{(bc + c - 1) \log c - c + 1}{bc^2(bc + c - 1)}.$
        When $c = 1$ then for all $b$, $\frac{\partial y}{\partial c} = 0$.
        As the denominator is positive for $b,c \geq 1$ it is sufficient to show that the numerator $z = (bc + c - 1) \log c - c + 1$ is non-negative.
        To show that $z \geq 0$, we consider the partial derivative again with respect to $c$:
        $\frac{\partial z}{\partial c} = (b + 1) \log c + b - \frac{1}{c}$
        For $b,c\ge 1$, we have that $b - \frac{1}{c} \geq 0$ and $(b+1)\log(c) \geq 0$.
        Since $\frac{\partial z}{\partial c} \geq 0$, we conclude that $z$ is increasing and non-negative which implies $y$ is also increasing and non-negative, for all $b, c \geq 1$. This proves the lemma.
    \end{proof}

    The proof of the next lemma follows the proof of Lemma 2.2 from  \cite{FominGLS16}, who proved it for $b=1$.

    \begin{lemma}
        \label{lem:orderComplexityEquality}
        Let $b, c \geq 1$, $n$ and $k \leq n$ be non-negative integers. Then, there exists $t \geq 0$ such that
        \[
            \frac{\binom{n}{t}}{\binom{k}{t}}   b^{n-t}   c^{k-t} = \left( 1 + b - \frac{1}{c} \right)^n n^{O(1)}
        \]
    \end{lemma}

    \begin{proof}
        We consider two cases. First suppose $k \leq \frac{n}{bc}$. Then for $t = 0$ the LHS (left-hand side) is at most $b^n c^k \leq b^n c^{\frac{n}{bc}} \leq \left( 1 + b - 1/c \right)^n$ by Lemma \ref{thm:case1inequality}. 
        Now if $k > \frac{n}{bc}$ then we rewrite the LHS as
        \[ 
            \frac{\binom{n}{t}}{\binom{k}{t}}   b^{n-t}   c^{k-t} 
            = \frac{\binom{n}{k}b^{n-k}}{\binom{n-t}{k-t} \left( \frac{1}{bc} \right)^{k-t}}
        \]
        Let us lower bound the denominator. For any $x \geq 0$ and an integer $m \geq 0$,
        \begin{equation}
            \label{eq:generatingSum}
            \sum_{i \geq 0} \binom{m + i}{i} x^i = \sum_{i \geq 0} \binom{m + i}{m} x^i = \frac{1}{(1-x)^{m+1}},
        \end{equation}
        by a known generating function.
        For $m = n - k$ and $x = \frac{1}{bc}$, the summand at $i = k - t$ equals the denominator $\binom{n-t}{k-t} \left( \frac{1}{bc} \right)^{k-t}$. Since $\frac{n}{k} < bc$ we have that $\frac{m+k}{k} < \frac{1}{x}$ and the terms of this sum decay exponentially for $i > k$. Thus, the maximum term $\frac{(m+i)  (m+i-1)  \ldots   (m+1)}{i  (i-1)   \ldots   1} x^i$ for this sum occurs for $i \leq k$, and its value is $\Omega \left( \left( \frac{1}{1-x} \right)^m \right)$ up to a lower order factor of $O(k)$. So by the binomial theorem the expression is at most
        \[ 
            \binom{n}{k}b^{n-k}(1-x)^{n-k} n^{O(1)} = \left( 1 + b - \frac{1}{c} \right)^n n^{O(1)} 
        \]
        Specifically, the maximum term for Equation \eqref{eq:generatingSum} occurs when $\frac{m + i}{i} = \frac{1}{x}$, that is when $\frac{n - t}{k - t} = cb$, and therefore, $t=\frac{cbk-n}{cb-1}$. 
    \end{proof}



    \begin{lemma}
        \label{lem:extensionToExtension}
        If there exist constants $b, c \geq 1$ and an algorithm for \phiext with running time $b^{n-|X|}  c^k   \bitsize^{O(1)}$ then there exists a randomized algorithm for \phiext with running time $\left( 1 + b - \frac{1}{c} \right)^{n - |X|}   \bitsize^{O(1)}$

    \end{lemma}

    \begin{proof}
        Our proof is similar to Lemma 2.1 in \cite{FominGLS16}. Let $\mathcal{B}$ be an algorithm for $\Phi$-\textsc{Extension} with running time $b^{n-|X|} c^k \bitsize^{O(1)}$. We now present a randomized algorithm $\mathcal{A}$, for the same problem for an input instance $(I,X,k')$ with $k' \leq k$.
        \begin{enumerate}
            \item Choose an integer $t \leq k'$ depending on $b, c, n, k'$ and $|X|$, the choice of which will be discussed later. Then select a random subset $Y$ of $U_I \backslash X$ of size $t$.
            \item Run Algorithm $\mathcal{B}$ on the instance $(I, X \cup Y, k' - t)$ and return the answer.
        \end{enumerate}
        Algorithm $\mathcal{A}$ has a running time upper bounded by $b^{n-|X|-t} c^{k'-t} \bitsize^{O(1)}$.
        Algorithm $\mathcal{A}$ returns yes for $(I,X,k')$ when $\mathcal{B}$ returns yes for $(I, X \cup Y, k' - t)$. In this case there exists a set $S \subseteq U_I \backslash (X \cup Y)$ of size at most $k'  - t \leq k - t$ such that $S \cup X \cup Y \in \mathcal{F}_I $. This, $Y \cup S$ witnesses that $(I, X, k)$ is indeed a yes-instance.

        Next we lower bound the probability that $\mathcal{A}$ returns yes if there exists a set $S \subseteq U_I \backslash X$ of size exactly $k'$ such that $X \cup S \in \mathcal{F}_I$. The algorithm $\mathcal{A}$ picks a set $Y$ of size $t$ at random from $U_I \backslash X$. There are $\binom{n - |X|}{t}$ possible choices for $Y$. If $\mathcal{A}$ picks one of the $\binom{k'}{t}$ subsets of $S$ as $Y$ then $\mathcal{A}$ returns yes. Thus, given that there exists a set $S \subseteq U_I \backslash X$ of size $k'$ such that $X \cup S \in \mathcal{F}_I$, we have that
        \[
            \Pr[ \mathcal{A} \text{ returns yes}] \geq \Pr[Y \subseteq S] = \binom{k'}{t} / \binom{n - |X|}{t}
        \]
        Let $p(k') =  \binom{k'}{t} / \binom{n - |X|}{t}$.
        For each $k' \in \{0 ,... , k\} $, our main algorithm runs $\mathcal{A}$ independently $\frac{1}{p(k')}$ times with parameter $k'$. The algorithm returns yes if any of the runs of $\mathcal{A}$ return yes. If $(I, X, k')$ is a yes-instance, then the main algorithm returns yes with probability at least 
        \[
            \min_{k' \leq k} \left\{ 1 - (1 - p(k'))^{\frac{1}{p(k')}} \right\} \geq 1 - \frac{1}{e} > \frac{1}{2}.
        \]
        Next we upper bound the running time of the main algorithm, which is
        
        \begin{align}
            \label{eq:kprimemaxinequality}
            \sum_{k' \leq k} \frac{1}{p(k')}   b^{n - |X| - t} c^{k' - t} \bitsize^{O(1)} &\leq
            \max_{k' \leq k} \frac{\binom{n - |X|}{t}}{\binom{k'}{t}}   b^{n - |X| - t} c^{k' - t} \bitsize^{O(1)} \\ &\leq
            \max_{k' \leq n - |X|} \frac{\binom{n - |X|}{t}}{\binom{k}{t}}   b^{n - |X| - t} c^{k - t} \bitsize^{O(1)}.
        \end{align}

        \noindent
        The choice of $t$ in algorithm $\mathcal{A}$ is chosen to minimize the value of $\frac{\binom{n - |X|}{t}}{\binom{k}{t}}   b^{n - |X| - t}   c^{k - t}$. For fixed $n$ and $|X|$ the running time of the algorithm is upper bounded by
        \begin{equation}
            \label{eq:runningTimeMinMax}
            \max_{0 \leq k \leq n - |X|} \left\{ 
                \min_{0 \leq t \leq k} \left\{
                    \frac{\binom{n - |X|}{t}}{\binom{k}{t}} b^{n - |X| - t} c^{k - t} \bitsize^{O(1)}
                \right\}
            \right\}
        \end{equation}
        By application of Lemma \ref{lem:orderComplexityEquality} we choose $t = \frac{cbk - (n - |X|)}{cb - 1}$ to obtain the upper bound 
        
        \[
            \left( 1 + b - \frac{1}{c} \right)^{n - |X|}   (n - |X|)^{O(1)},
        \]
        
        which, combined with $n < N$, completes the proof.
    \end{proof}

    \noindent
    Running algorithm $\mathcal{A}$ with $X = \emptyset$ and for each value of $k \in \{0,....,n\}$ results in an algorithm for  \phisub with running time $\left( 1 + b - \frac{1}{c} \right)^n   \bitsize^{O(1)}$, proving Theorem \ref{thm:main1}.

\section{Derandomization}
    In this section we prove Theorem \ref{thm:main2}, by derandomizing the algorithm in Theorem \ref{thm:main1}.

    \thmMainTwo*

    \noindent
    Given a set $U$ and an integer $q \leq |U|$ let $\binom{U}{q}$ represent the set of sets which contain $q$ elements of $U$.
    From \cite{FominGLS16} we define a pseudo-random object, the \textit{set-inclusion-family}, as well as an almost optimal sub-exponential construction of these objects.

    \begin{definition}
        Let $U$ be a universe of size $n$ and let $0 \leq q \leq p \leq n$. A family $\mathcal{C} \subseteq \binom{U}{q}$ is an $(n, p, q)$\text{-set-inclusion family}, if for every set $S \in \binom{U}{p}$, there is a set $Y \in \mathcal{C}$ such that $Y \subseteq S$.
    \end{definition}

    \noindent
    Let $\kappa(n,p,q) = \binom{n}{q} / \binom{p}{q}$. We also make use of the following theorem.

    \begin{theorem}[\cite{FominGLS16}]
        \label{thm:FGLSderandom}
        There is an algorithm that given $n, p$ and $q$ outputs an $(n,p,q)$-set-inclusion-family $\mathcal{C}$ of size at most $\kappa(n,p,q) \cdot 2^{o(n)}$ in time $\kappa (n,p,q) \cdot 2^{o(n)}$.
    \end{theorem}

    We are now ready to prove Lemma \ref{lem:deterministicExtension}, by a very similar proof to Lemma \ref{lem:extensionToExtension}.

    \begin{lemma}
        \label{lem:deterministicExtension}
        If there exists constants $b, c \geq 1$ and an algorithm for \phiext with running time $b^{n-|X|} c^k  \bitsize^{O(1)}$ then there exists a deterministic algorithm for \phiext with running time $\left( 1 + b - \frac{1}{c} \right)^{n - |X|} \cdot 2^{o(n)} \cdot \bitsize^{O(1)}$.
    \end{lemma}

    \begin{proof}
    Let $\mathcal{B}$ be an algorithm for \phiext with running time $b^{n-|X|} c^k \bitsize^{O(1)}$. We can then adapt Algorithm $\mathcal{A}$ from the proof of Lemma \ref{lem:extensionToExtension}. Let $\mathcal{A'}$ be a new algorithm which has an input instance $(I, X, k')$ with $k' \leq k$. Choose $t = \frac{cbk' - (n - |X|)}{cb - 1}$.

    \begin{enumerate}
        \item Compute an $(n - |X|, k', t)$-set-inclusion-family $\mathcal{C}$ using the algorithm from Theorem \ref{thm:FGLSderandom} of size at most $\kappa(n-|X|, k', t) \cdot 2^{o(n)}$, in $\kappa(n-|X|, k', t) \cdot 2^{o(n)}$ time.
        \item For each set $Y$ in the set-inclusion-family $\mathcal{C}$ run algorithm $\mathcal{B}$ on the instance $(I, X \cup Y, k' - t)$ and return \textsc{Yes} of at least one returns \textsc{Yes} and \textsc{No} otherwise.
    \end{enumerate}
    \noindent
    The running time of $\mathcal{A'}$ is upper bounded by $\kappa(n - |X|, k', t) \cdot 2^{o(n)} \cdot b^{n-|X|-t} c^{k'-t} \bitsize^{O(1)}$, a term encountered in Equation \ref{eq:kprimemaxinequality} with a new subexponential factor in $n$,
    \[
        \max_{k' \leq k} \frac{\binom{n - |X|}{t}}{\binom{k'}{t}} \cdot b^{n-|X|-t} c^{k'-t} \bitsize^{O(1)} \cdot 2^{o(n)}.
    \]
    From here the proof follows that of Lemma \ref{lem:extensionToExtension}.
    \end{proof}
    The proof of Theorem \ref{thm:main2} follows by inclusion of the factor $2^{o(n)}$.

\section{Enumeration}
    We now proceed to prove Theorem \ref{thm:main3}, and \ref{thm:main4} on combinatorial upper bounds and enumeration algorithms.
    Consider the following random process.

    \begin{enumerate}
        \item Choose an integer $t$ based on $b, c, n$ and $k$, then randomly sample a subset $X$ of size $t$ from $U_I$.
        \item Uniformly at random pick a set $S$ from $\cF_{I, X}^{k-t}$, and output $W = X \cup S$. In the special case where $\cF_{I,X}^{k-t}$ is empty output the empty set.
    \end{enumerate}


    \thmMainThree*

    \begin{proof}
        Let $I$ be an instance, $k \leq n$. We will prove that the number of sets in $\cF_I$ of size exactly $k$ is upper bounded by $|\cF_I| \leq \left( 1 + b - \frac{1}{c} \right)^{n} n^{O(1)}$, where $k$ is chosen arbitrarily. We follow the random process described above, which picks a set $W$ of size $k$ from $\cF_I$.

        For each set $Z \in \cF_I$ of size exactly $k$, let $E_Z$ denote the event that the set $W$ output in step 2 is equal to $Z$. We then have the following lower bound on the probability of the event $E_Z$:

        \begin{align*} 
            \Pr[E_Z] &= \Pr[X \subseteq Z \land S = Z \backslash X] \\
                     &= \Pr[X \subseteq Z] \times \Pr[S = Z \backslash Z | X \subseteq Z] \\
                     &= \frac{\binom{k}{t}}{\binom{n}{t}} \cdot\frac{1}{\left|\cF_{I,X}^{k-t}\right|}
        \end{align*}
        
        Since $\Phi$ is $(b,c)$-uniform then $\left|\cF_{I,X}^{k-t}\right| \leq b^{n - |X|}c^k n^{O(1)}$ and $X$ is selected such that $|X| = t$, this results in the lower bound
        \[
            \Pr[E_Z] \geq \frac{\binom{k}{t}}{\binom{n}{t}} b^{-(n - t)}c^{-(k-t)} n^{-O(1)}.
        \]
        A choice of $t$ is made to minimize the lower bound, and this choice is given by Lemma \ref{lem:orderComplexityEquality} which states that for every $k \leq n$ there exists a $t \leq k$ such that we obtain a new lower bound
        \[
            \Pr[E_Z] \geq \left( 1 + b - \frac{1}{c} \right)^{-n} \cdot n^{O(1)}
        \]
        for every $Z \in \cF_I$ of size $k$. For every individual set $Z \in \cF_I$, the event $E_Z$ occurs disjointly, and we have that
        $\sum_{Z \in \cF_I, |Z| = k} \Pr[E_Z] \leq 1$.
        This fact with the lower bound of $\Pr[E_Z]$ implies an upper bound on the number of sets in $\cF_I$ of $(1 + b - \frac{1}{c})^n n^{O(1)}$, completing the proof.

    \end{proof}

    \thmMainFour*

    \begin{proof}
        We alter the random process used to prove Theorem \ref{thm:main3} to a deterministic one:
        \begin{enumerate}
            \item Construct a $(n, k, t)$-set inclusion family $\mathcal{C}$ using Theorem 6 from \cite{FominGLS16}. Loop over $X \in \mathcal{C}$.
            \item For each $X \in \mathcal{C}$, loop over all sets $S \in \cF_{I,X}^{k-t}$.
        \end{enumerate}
        Then we output $W = X \cup S$ from these two loops. Looping over $\mathcal{C}$ instead of random sampling for $X$ incurs a $2^{o(n)}$ overhead in the running time. As $\Phi$ is efficiently $(b,c)$-uniform, the inner loop requires $(1 + b - \frac{1}{c})^n N^{O(1)}$ time. In order to avoid enumerating duplicates, we save the sets which have been output in a trie and check first in linear time if a set has already been output.
        The product of the running times for these two nested loops results in the running time claimed by the theorem statement.
    \end{proof}

\section{Feedback Vertex Set}

    \label{study:fvs}

    First we describe the extension variant of \textsc{Feedback Vertex Set}

    \pbDef{\fvsext}{A graph $G = (V,E)$, vertex subset $X \subseteq V$ and an integer $k$}
    {Does there exist subset $S \subseteq V \backslash X$ such that $S \cup X$ is a FVS and $|S| \leq k$?}

    \noindent
    Instead of directly finding the feedback vertex set in a graph, we present algorithm $\mif(G,F,k)$ \cite{fomin2008minimum} which computes for a given graph $G$ and an acylic set $F$ the maximum size of an induced forest $F'$ containing $F$ with $|F'| \geq n - k$. This means that $G - F$ is a minimal feedback vertex set of size at most $k$. This algorithm can easily be turned into an algorithm computing at least one such set. 

    During the execution of $\mif$ one vertex $t \in F$ is called an \textit{active vertex}. Algorithm $\mif$ then branches on a chosen neighbor of $t$. Let $v \in N(t)$. Let $k$ be the set of all vertices of $F \backslash \{t\}$ that are adjacent to $v$ and parameter $k$ which represents a bound on the size of the feedback vertex set.

    As well as describing the algorithm we simultaneously perform the running time analysis which uses the Measure and Conquer framework and Lemma \ref{lem:11} at its core.

    \begin{lemma}{\textsc{\cite{gaspers2010exponential}}}
        \label{lem:11}
        Let $A$ be an algorithm for a problem $P$, $B$ be an algorithm for a class $C$ of instances of $P$, $c \geq 0$ and $r > 1$ be constants, and $\mu(\cdot), \mu_B(\cdot), \eta(\cdot)$ be measures for $P$, such that for any input instance $I$ from $C$, $\mu_B(\cdot) \leq \mu (I)$, and for any input instance $I$, $A$ either solves $P$ on $I \in C$ by invoking $B$ with running time $O(\eta(I)^{c+1}r^{\mu B(I)})$, or reduces $I$ to $k$ instances $I_1,...,I_k$, solves these recursively, and combines their solutions to solve $I$, using time $O(\eta(I)^c)$ for the reduction and combination steps (but not the recursive solves),
        \begin{align}
            (\forall i) \quad \eta(I_i) \leq \eta (I) - 1, \text{ \quad and \quad }
            \sum_{i=1}^k r^{\mu(I_i)} \leq r^{\mu(I)}.
        \end{align}
        Then $A$ solves any instance $I$ in time $O(\eta(I)^{c+1} r^{\mu(I)})$.
    \end{lemma}

    \noindent
    \textit{Branching constraints} of the form  $\sum_{i = 1}^j 2^{-\delta_i} \leq 1$ are given as branching vectors $(\delta_1, ..., \delta_j)$.

    \subsubsection{Measure}

    To upper bound the exponential time complexity of the algorithm \texttt{mif} we use the measure
    \[
        \mu = |N(t) \backslash F|w_1 + |V \backslash (F \cup N(t))|w_2 + k \cdot w_k.
    \]
    In other words, each vertex in $F$ has weight 0, each vertex in $N(t)$ has weight $w_1$, each other vertex has weight $w_2$ and each unit of budget for the feedback vertex set has weight $w_k$, in the measure with an active vertex $t$.

    \subsubsection{Algorithm}

    The description of $\mif$ consists of a sequence of cases and subcases. The first case which applies is used, and inside a given case the hypotheses of all previous cases are assumed to be false. Preprocessing procedures come before main procedures. 

    \paragraph*{Preprocessing}

    \begin{enumerate}
        \item If $G$ consists of $j \geq 2$ connected components $G_1, G_2, ... ,G_j$, then the algorithm is called on each component. For $F_i = G_i \cap F$ for all $i \in \{1, 2, ..., j\}$ and $\sum_{i = 1}^j k_i \leq k $ then
        \[
            \mif(G, F, k) = \sum_{i=1}^{j} \mif(G_i, F_i, k_i)
        \]
            

        \item If $F$ is not independent, then apply operation $\text{Id}^*(T, v_T)$ on an arbitrary non-trivial component $T$ of $F$. If $T$ contains the active vertex then $v_T$ becomes active. Let $G'$ be the resulting graph and let $F'$ be the set of vertices of $G'$ obtained from $F$. Then
        \[
            \mif(G, F, k) = \mif(G', F', k) + |T| - 1
        \]

    \end{enumerate}

    \paragraph*{Main Procedures}

    \begin{enumerate}
        \item If $k < 0$ then

        \[
            \mif(G, F, k) = 0.
        \]

        \item If $F = \emptyset$ and $\Delta(G) \leq 1$ then $\cM_G(F) = \{V\}$ and

        \[
            \mif(G, F, k) = |V|.
        \]

        \item If $F = \emptyset$ and $\Delta(G) \geq 2$ then the algorithm chooses a vertex $t \in G$ of degree at least 2. Then $t$ is either contained in a maximum induced forest or not. The algorithm branches on two subproblems and returns the maximum:
        \[
            \mif(G, F, k) = \max \{ \mif(G, F \cup \{t\}, k), \mif(G \backslash \{t\}, F, k-1) \}.
        \]
        The first branch reduces the weight of $t$ to zero, as it is in $F$, and at least 2 neighbors have a reduced degree from $w_2 $ to $ w_1$. In the second branch we remove $t$ from the graph, meaning it will be in the feedback vertex set. We thus also gain a reduction of $w_k$ in the measure. Hence this rule induces the branching constraint
        \[
            (w_2 + 2(w_2 - w_1), w_2 + w_k).
        \]

        \item If $F$ contains no active vertex then choose an arbitrary vertex $t \in F$ as an active vertex. Denote the active vertex by $t$ from now on.

        \item If $V \backslash F = N(t)$ then the algorithm constructs the graph $H$ from Proposition \ref{prop:3} and computes a maximum independent set $I$ in $G$ of maximum size $n - k$. Then
        \[
            \mif(G, F, k) = |F| + |I|.
        \]

        \item If there is $v \in N(t)$ with $\text{gd}(v) \leq 1$ then add $v$ to $F$.
        \[
            \mif(G,F,k) = \mif(G, F \cup \{v\}, k).
        \]

        \item If there is $v \in N(t)$ with $\text{gd} (v) \geq 4$ then either add $v$ to $F$ or remove $v$ from $G$.
        \[
            \mif(G, F, k) = \max \{ \mif(G, F\cup\{v\}, k), \mif(G \backslash \{v\}, F, k-1) \}.
        \]
        The first case adds $v$ to $F$ reducing the measure by $w_1$, and a minimum of $4(w_2 - w_1)$ for each of the generalized neighbors. The other case removes $v$ this decreasing the measure by $w_k$. Hence this rule induces the branching constraint
        \[
            (w_1 + 4(w_2 - w_1), w_1 + w_k).
        \]

        \item \label{case:fvs2neg} If there is $v \in N(t)$ with $\text{gd}(v) = 2$ then denote its generalized neighbors by $u_1$ and $u_2$. Either add $v$ to $F$ or remove $v$ from $G$ but add $u_1$ and $u_2$ to $F$. If adding $u_1$ and $u_2$ to $F$ induces a cycle, we just ignore the last branch.
        \[
            \mif(G, F, k) = \max \{ \mif(G, F\cup\{v\}, k), \mif(G \backslash \{v\}, F \cup \{u_1, u_2\}, k-1) \}.
        \]
        Let $i \in \{0, 1, 2\}$ be the number of vertices adjacent to $v$ with weight $w_2$. The first case adds $v$ to $F$, and hence all $i$ $w_2$-weight neighbors of $v$ reduce to $w_1$, and the other $2-i$ vertices of weight $w_1$ induce a cycle, hence we remove them from $G$ and reduce the measure by $(2-i)w_k$. The second case removes $v$ and adds both $u_1$ and $u_2$ to $F$. This causes a reduction of $i w_2$ for the relevant vertices and $(2-i)w_1$ for the other vertices, and a single $w_k$ reduction due to the removal of $v$. This rule induces the branching constraint
        \[
            (w_1 + i (w_2 - w_1) + (2 - i)w_1 + (2-i)w_k, w_1 + iw_2 + (2-i)w_1 + w_k).
        \]

        \item If all vertices in $N(t)$ have exactly three generalized neighbors then at least one of these vertices must have a generalized neighbor outside $N(t)$, since the graph is connected and the condition of the case Main 6 does not hold. Denote such a vertex by $v$ and its generalized neighbors by $u_1$, $u_2$ and $u_3$ in such a way that $u_1 \not\in N(t)$. Then we either add $v$ to $F$; or remove $v$ from $G$ but add $u_1$ to $F$; or remove $v$ and $u_1$ from $G$ and add $u_2$ and $u_3$ to $F$. Similar to the previous case, if adding $u_2$ and $u_3$ to $F$ induces a cycle, we just ignore the last branch.
        \begin{align*}
            \mif(G,F) = \max \{ &\mif(G, F \cup \{v\}, k), 
                                 \mif(G \backslash \{v\}, F \cup \{u_1\}, k - 1), \\
                                &\mif(G \backslash \{v, u_1\}, F \cup \{u_2, u_3\}, k - 2)\}.
        \end{align*}
        Let $i \in \{1, 2, 3\}$ be the number of vertices adjacent to $v$ with weight $w_2$. The first and last cases are analogous to the analysis done in Main \ref{case:fvs2neg}. The second case removes $v$ from the forest hence adding it to the minimum feedback vertex set and reducing the measure by $w_1 + w_k$. A reduction of $w_2$ is gained by adding $u_1$ to $F$. Then this rule induces the branching constraint
        \begin{align*}
            (w_1 + i(w_2-w_1) + (3-i) w_1 + (3-i) w_k, 
             w_1 + w_2 + w_k, 
             w_1 + i w_2 + (3-i) w_1 + 2 w_k ).
        \end{align*}

    \end{enumerate}

    \subsubsection{Results}

    \begin{theorem}
        Let $G$ be a graph on n vertices. Then a minimal feedback vertex set in $G$ can be found in time \rtfvs.
    \end{theorem}


    \begin{proof}
        Using the algorithm above along with the measure $\mu$, the following values of weights can be shown to satisfy all the branching vector constraints generated above.
        \[
            w_1 = 0.2775 \quad \quad \quad \quad
            w_2 = 0.6250 \quad \quad \quad \quad
            w_k = 0.2680 \quad \quad \quad \quad
        \]
        These weights result in an upper bound for the running time of \mif as $O(1.5422^n \cdot 1.2041^k)$ for computing a maximally induced forest of size a least $n - k$, and hence we have the running time for \fvsext of $O(1.5422^{n-|X|} \cdot 1.2041^k)$. By Theorem \ref{thm:main2} this results in a \rtfvs algorithm for computing a minimal feedback vertex set.
    \end{proof}

\section{Minimal Vertex Covers}

    \begin{restatable}{theorem}{thmMVCresult}
        \label{thm:mvc}
        Let $\sergeGamma$ be a constant with $0.169925 \approx 2 \log_2 3-3 \le \sergeGamma \le 1$.
        For every $n\ge k\ge 0$, and every graph $G$ on $n$ vertices, the number of minimal vertex covers of size at most $k$ of $G$ is at most
        $2^{\sergeBeta n + \sergeGamma k}$, where $\sergeBeta = (1-\sergeGamma)/2$.
    \end{restatable}

    \begin{proof}
        The proof is by induction on $n$.
        For the base case, a graph on at most one vertex has one minimal vertex cover -- the empty set -- and $2^{\sergeBeta n + \sergeGamma k} \ge 1$ since $\sergeBeta n + \sergeGamma k\ge 0$.
        
        Suppose the statement holds for graphs with fewer than $n$ vertices.
        We will repeatedly use the observation that for every vertex $v$, no minimal vertex cover of $G$ contains $N[v]$.
        Let $v$ be a vertex of minimum degree in $G$.
        
        If $v$ has degree $0$, then no minimal vertex cover contains $v$. Thus, $G$ has as many minimal vertex covers as $G-v$. The number of minimal vertex covers of $G$ is therefore upper bounded by
        \begin{align*}
         2^{\sergeBeta (n-1) + \sergeGamma k} \le 2^{\sergeBeta n + \sergeGamma k}.
        \end{align*}
        
        If $v$ has degree $1$, then every minimal vertex cover either excludes $v$ but includes its neighbor $u$, or it includes $v$ but excludes $u$.
        The number of minimal vertex covers of $G$ is therefore upper bounded by
        \begin{align*}
        2\cdot 2^{\sergeBeta (n-2) + \sergeGamma (k-1)}  \le 2^{\sergeBeta n + \sergeGamma k -(2\sergeBeta +\sergeGamma) +1} = 2^{\sergeBeta n + \sergeGamma k}
        \end{align*}
        since $2\sergeBeta +\sergeGamma = 1$.
        
        If $v$ has degree $2$, then every minimal vertex cover excludes a vertex among $N[v]$, but includes its neighbors who all have degree at least $2$.
        The number of minimal vertex covers of $G$ is therefore upper bounded by
        \begin{align*}
        3\cdot 2^{\sergeBeta (n-3) + \sergeGamma (k-2)}  \le 2^{\sergeBeta n + \sergeGamma k -(3\sergeBeta +2\sergeGamma) +\log_2 3} \le 2^{\sergeBeta n + \sergeGamma k}
        \end{align*}
        since $3\sergeBeta +2\sergeGamma = \frac{3+\sergeGamma}{2}\ge \log_2 3$.
        
        If $v$ has degree at least $3$, every minimal vertex cover includes $v$ or excludes $v$ but includes all its neighbors.
        The number of minimal vertex covers of $G$ is therefore upper bounded by
        \begin{align*}
        2^{\sergeBeta (n-1) + \sergeGamma (k-1)} + 2^{\sergeBeta (n-4) + \sergeGamma (k-3)}  &\le 2^{\sergeBeta n + \sergeGamma k} \cdot (2^{-\sergeBeta-\sergeGamma}+2^{-4\sergeBeta -3\sergeGamma})\\
        &= 2^{\sergeBeta n + \sergeGamma k} \cdot (2^{-\frac{1+\sergeGamma}{2}}+2^{-2-\sergeGamma}) \le 2^{\sergeBeta n + \sergeGamma k}
        \end{align*}
        since $2^{-\frac{1+\sergeGamma}{2}}+2^{-2-\sergeGamma} \le 0.89$.
    \end{proof}

    The upper bound of Theorem \ref{thm:mvc} is tight for every $\sergeGamma$ within the constraints of the theorem, as shown by $1$-regular graphs with $k=n/2$.
    For $\sergeGamma = 2 \log_2 3-3$, the disjoint union of triangles also matches the upper bound for $k=2n/3$.




    We note that the proof of Theorem \ref{thm:mvc} can easily be turned into an algorithm enumerating all minimal vertex covers of $G$ in time $2^{\sergeBeta n + \sergeGamma k} n^{O(1)}$. Alternatively, a polynomial-delay algorithm, such as the one by \cite{JohnsonP88}, could be used for the enumeration.

\section{Minimal Feedback Vertex Sets}

    In this section, we apply our framework to enumerating all minimal feedback vertex sets of an undirected graph on $n$ vertices.
    We will modify the algorithm from \cite{fomin2008minimum}, and conduct a multivariate branching analysis. 
    When combined with Theorem \ref{thm:main4} we obtain an algorithm for enumerating all minimal feedback vertex sets in time \rtenumfvs.

    \subsection{Measure}


    Following \cite{fomin2008minimum}, we show that for any acyclic subset $F$ of $G = (V,E)$, $|\cM_G(\emptyset)| \leq $ \rtenumfvs. 
    We assume $F$ is independent by Proposition \ref{prop:1}. For a graph $G$, an independent set $F$, and an \emph{active vertex} $t \in F$, we use the measure:
    \[
        \mu(G,F,t) = |A| \wA + |N(t) \bs (F \cup A)| \wB + |V \bs (F \cup N(t))| \wBeta + k \cdot \wK 
    \]
    where the set $A \subseteq N(t) \bs F$ consists of vertices which have generalized degree at least 3. We apply positive weights $\wA$, $\wB$ and $\wBeta$ to the three sets defined, with $0 \leq \wA \leq \wB \leq \wBeta$. A weight of $\wK$ is applied to the each vertex in the feedback vertex set.

    \subsection{Algorithm}

    Similar to \textsc{Feedback Vertex Set} in Subsection \ref{study:fvs}, we perform an algorithm description and a running time analysis using a Measure and Conquer framework simultaneously. 

    Let $f(G,F, k)= |M_G (F )|$ be the number of maximal induced forests containing $F$ of size at least $n - k$.
    Let $f (\mu, k)$ be a maximum $f(G,F, k)$ among all four-tuples $(G,F,t, k)$ of measure at most $\mu$.

    For the algorithm denote $t \in F$ to be the active vertex. If $F \neq \emptyset$ contains no active vertex then we choose an arbitrary vertex as active, reducing the measure.


    \subsubsection*{Cases}

    \begin{enumerate}
        \item If $k < 0$ then $f(G, F, k) = 0$.

        \item If $k = 0$ then $f(G, F, k) = 1$ if $G = F$ otherwise $f(G,F,k) = 0$. 




        \item \emph{$F = \emptyset$.} If $\Delta(G) \leq 1$ then $\cM_G(F) = \{V\}$ so $f(G,F,k) = 1$. Otherwise choose an active vertex $t \in V$ of degree at least 2. Every maximal forest either contains $t$ or does not, meaning that the number of maximal forests is
        \[
            f(G, \{t\}, k) + f(G \bs \{t\}, \emptyset, k - 1)
        \]
        which results in the branching vector
        \[
            (\wBeta + 2(\wBeta - \wB), \wBeta + \wK).
        \]
        From now on, denote $t \in F$ as the active vertex. Let $G_t = (V_t, E_t)$ be the connected component of $G$ which contains $t$. 
        
        \item \emph{$V_t \bs F = N(t)$.} By Proposition \ref{prop:3}, $f(\mu, k)$ is equal to the number of maximal independent sets in the graph $H$ of size at least $n - k$. By Theorem \ref{thm:mvc}, we have an upper bound on the number of minimal vertex covers of size at most $k$ giving us an upper bound also on the maximal independent sets of size at least $n - k$. We ensure that this computation is not worse than that of enumerating feedback vertex sets. 
        \[
            f(\mu, k) \leq 2^{\sergeBeta n + \sergeGamma k} \leq 2^\mu
        \]
        for $2 \log_2 3 - 3 \leq \sergeGamma \leq 1$ and $\sergeBeta = (1 - \sergeGamma)/2$.

        \item gd\emph{$(v) = 0$.} In this case every forest $X \in \cM_{G}(F)$ contains $v$ thus
        \[
            f(G, F, k) = f(G, F \cup \{c\}, k)
        \]
        which does not induce a branching vector.

        From this point on, pick a vertex $v \in N(t)$. If there is no such vertex then $t$ is no longer an active vertex and if $F \neq \emptyset$ then we choose an arbitrary vertex in $F$ as active.

        \item gd\emph{$(v) = 1$.} In this case every forest $X \in \cM_{G}(F)$ either contains $v$ or does not contain $v$ and contains its generalized neighbor $u$. This means that the number of maximal induced forests is at most
        \[
            f(G, F \cup \{v\}, k) + f(G \bs \{v\}, F \cup \{u\}, k-1).
        \]
        If we have that $u \in N(t)$, in the worst case we have the branching vector
        \[
            (\wB + (\wA + \wK), \wB + \wA + \wK)
        \]

        otherwise if $u \not \in N(t)$ we have the branching vector 
        \[
            (\wB + (\wBeta - \wB), \wB + \wBeta + \wK).
        \]

        \item gd\emph{$(v) = 2$.} Denote the generalized neighbors of $v$ by $u_1$ and $u_2$, and assume that $u_1 \not\in N(t)$. If $u_2 \in N(t)$ and $v$ belongs to a maximal induced forest $X$ then $u_2$ does not belong to $X$. Then every forest $X$ from $\cM_G(F)$ satisfies one of the following conditions:

        \begin{itemize}
            \item either $X$ contains $v$, but not $u_2$,
            \item or $X$ does not contain $v$, and contains $u_1$,
            \item or $X$ does not contain $v$ and $u_1$ but contains $u_2$.
        \end{itemize}

        So the number of maximal forests is at most
        \[
            f(G \bs \{u_2\}, F \cup \{v\}, k) + f(G \bs \{v\}, F \cup \{u_1\}, k-1) + f(G \bs \{v, u_1\}, F \cup \{u_2\}, k-2).
        \]

        In the worst case, where $u_2$ has weight $\wA$, then this results in the branching vector
        \[
            (\wB + (\wBeta - \wB) + \wA + \wK, \wB + \wBeta + \wK, \wB + \wBeta + \wA + 2\wK).
        \]

        However, if $u_1, u_2 \not \in N(t)$, assume gd$(u_1) \leq$ gd$(u_2)$. If not, swap $u_1$ and $u_2$. We consider new subcases and rules based on $d(u_1)$, the generalized degree of the vertex, and the structure of the local graph near the vertex $u_1$. If gd($u_1$) = 2, let $x_1$ and $x_2$ be the two generalized neighbors of $u_1$. 

        The weights of $v, u_1, u_2$ are $\wB, \wBeta, \wBeta$ respectively. We also note that when $v$ is selected, $u_1$ and $u_2$, if not removed from the graph or already considered in the branching analysis, will result in a reduction in measure of at least $(\wBeta - \wB)$ for each of $u_1$ and $u_2$. The branching analysis below has different weights for $x_1$ and $x_2$ depending on the subcase of the algorithm which is applied.

        \begin{enumerate}
            \item gd$(u_1) = 0$. Since every maximal forest $X \in \cM_G(F)$ will contain $u_1$, then $f(G, F, k) = f(G, F \cup \{u_1\}, k)$. This does not induce a branching vector.

            \item gd$(u_1) = 1$. Let the generalized neighbor of $u_1$ be $x$. Then every forest $X$ from $\cM_G(F)$ satisfies one of the following conditions:
            \begin{itemize}
                \item either $X$ contains $v$;
                \item or $X$ does not contain $v$ but contains $u_1$;
                \item or $X$ does not contain $v$ and $u_1$ but contains $u_2$ and $x$
            \end{itemize}
            which means the number of maximal induced forests is at most
            \[
                f(G, F \cup \{v\}, k) + f(G \bs \{v\}, F \cup \{u_1\}, k-1) + f(G \bs \{v, u_1\}, F \cup \{u_2, x\}, k - 2)
            \]
            which with worst case weights will induce the branching vector
            \[
                (\wB + 2(\wBeta - \wB), \wB + \wBeta + \wK, \wB + 2 \wBeta + 2 \wK + \wA).
            \]
            
            \item gd$(u_1) = 2$, $x_1 \in N(t)$ and $x_1$ is generalized neighbor of $u_1$ and $u_2$. If swapping $x_1$ and $x_2$ results in this case occurring, then do so.
            Then every forest $X$ from $\cM_G(F)$ satisfies one of the following conditions:
            \begin{itemize}
                \item either $X$ contains $v$ and $x_1$, but does not contain $u_1$ and $u_2$;
                \item or $X$ contains $v$ and not $x_1$;
                \item or$X$ does not contain $v$ and contains $u_1$;
                \item or$X$ does not contain $v$ and $u_1$ but contains $u_2$.
            \end{itemize}

            This means that the number of maximal induced forests is at most
            \begin{align*}
                f(G \bs \{u_1, u_2\}, F \cup \{v, x_1\}, k - 2) + &
                f(G \bs \{x_1\}, F \cup \{v\}, k - 1) + \\
                f(G \bs \{v\}, F \cup \{u_1\}, k- 1) + &
                f(G \bs \{v, u_1\}, F \cup \{u_2\}, k - 2).
            \end{align*}


            Since it is possible that $x_1$ has weight $\wA$ then in the worst case then this results in the branching vector
            \[
                (\wB + \wA + 2\wBeta + 2\wK, \wB + 2(\wBeta - \wB) + \wA + \wK, \wB + \wBeta + \wK, \wB + 2\wBeta + 2\wK).
            \]

            \item gd$(u_1) = 2$, $x_1 \in N(t)$ and $x_2 \in N(t)$. 
            hen every forest $X$ from $\cM_G(F)$ satisfies one of the following conditions:
            \begin{itemize}
                \item either $X$ contains $v$ and $u_1$, but does not contain $x_1$ and $x_2$;
                \item or $X$ contains $v$ and not $u_1$;
                \item or $X$ does not contain $v$ and contains $u_1$;
                \item or $X$ does not contain $v$ and $u_1$ but contains $u_2$.
            \end{itemize}

            This means that the number of maximal induced forests is at most
            \begin{align*}
                f(G \bs \{x_1, x_2\}, F \cup \{v, u_1\}, k - 2) + &
                f(G \bs \{u_1\}, F \cup \{v\}, k - 1) + \\
                f(G \bs \{v\}, F \cup \{u_1\}, k- 1) + &
                f(G \bs \{v, u_1\}, F \cup \{u_2\}, k - 2).
            \end{align*}
            
            In the worst case we have $x_1, x_2$ obtaining a weight of $\wA$, the branching vector is
            \[
                (\wB + \wBeta + 2\wA + 2\wK + (\wBeta - \wB), \wB + \wBeta + \wK + (\wBeta - \wB), \wB + \wBeta + \wK, \wB + 2\wBeta+ 2\wK).
            \]
            
            \item gd$(u_1) = 2$, and previous subcases do not apply. At least one of $x_1$ and $x_2$ has weight $\wBeta$ otherwise we would be in case (d). Let $x_1$ have weight $\wBeta$, and if not we can swap $x_1$ and $x_2$.
            Then every forest $X$ from $\cM_G(F)$ satisfies one of the following conditions:
            \begin{itemize}
                \item either $X$ contains $v$;
                \item or $X$ does not contain $v$ and contains $u_1$;
                \item or $X$ does not contain $v$ and $u_1$ but contains $u_2$ and $x_1$;
                \item or $X$ does not contain $v$, $u_1$ and $x_1$ but contains $u_2$ and $x_2$.
            \end{itemize}

            The number of maximal induced forests is thus at most
            \begin{align*}
                f(G, F \cup \{v\}, k) &+ f(G \bs \{v\}, F \cup \{u_1\}, k - 1) + \\
                f(G \bs \{v, u_1\}, F \cup \{u_2, x_1\}, k - 2) &+ f(G \bs \{v, u_1, x_1\}, F \cup \{u_2, x_2\}, k - 3).
            \end{align*}

            We now consider the weight of $x_2$, which is only ever selected into the forest.
            If $x_2$ is of weight $\wBeta$ the measure reduces by $\wBeta$.
            If $x_2$ is of weight $\wB$ since $x_2$ doesn't have both $u_1$ and $u_2$ as generalized neighbors due to case (c), in the worst case when $x_2$ is selected we also obtain a $(\wBeta - \wB)$ reduction.
            If $x_2$ is of weight $\wA$, we now have at least 2 unaccounted generalized neighbors which obtain a $2(\wBeta - \wB)$ reduction.

            We induce the following constraints to simplify the analysis so that in the worst case, we obtain a reduction of $\wBeta$ whenever $x_2$ is selected
            \[
                \wBeta \leq \wA + 2(\wBeta - \wB) \text{ and } \wBeta \leq \wB + (\wBeta - \wB)
            \]

            which results in the branching vector
            \[
                (\wB + 2(\wBeta - \wB), \wB + \wBeta + \wK, \wB + 3\wBeta + 2\wK, \wB + 4\wBeta + 3\wK).
            \]

            \item gd$(u_1) \geq 3$. This means that gd$(u_2) \geq 3$ as well. 
            hen every forest $X$ from $\cM_G(F)$ satisfies one of the following conditions:
            \begin{itemize}
                \item either $X$ contains $v$;
                \item or $X$ does not contain $v$ and contains $u_1$;
                \item or $X$ does not contain $v$ and $u_1$ but contains $u_2$.
            \end{itemize}
            This means the number of maximal induced forests is at most
            \[
                f(G, F \cup \{v\}, k) + f(G \bs \{v\}, F \cup \{u_1\}, k - 1) + f(G \bs \{v, u_1\}, F \cup \{u_2\}, k - 2).
            \]

            Since both $u_1$ and $u_2$ are generalized neighbors of $v$ with gd $\geq 3$ meaning they both obtain a weight of $\wA$ when $v$ is selected into the forest $X$. This establishes the branching vector
            \[
                (\wB + 2(\wBeta - \wA), \wB + \wBeta + \wK, \wB + 2\wBeta + 2\wK).
            \]

        \end{enumerate}

        \item gd\emph{$(v) = 3$.} Denote the generalized neighbors of $v$ by $u_1, u_2$ and $u_3$ according to the rule that $u_j \not \in N(t)$ and $u_k \in N(t)$ if and only if $j < k$. 

        Let $i$ be the number of generalized neighbors of $v$ that are not adjacent to $t$. For $i = 1,2$ we have that every forest $X$ from $\cM_G(F)$ satisfies one of the following conditions:
        \begin{itemize}
            \item either $X$ contains $v$;
            \item or $X$ does not contain $v$
        \end{itemize}
        meaning the number of maximal induced forests is at most
        \[
            f(G, F \cup \{v\}, k) + f(G \bs \{v\}, F, k-1).
        \]

        Each of the $3 - i$ generalized neighbors which are a neighbor of $t$ induces a cycle when $v$ is selected so we instead remove the vertex. This results in the branching vector
        \[
            (\wA + i(\wBeta - \wB) + (3 - i)(\wB + \wK), \wA + \wK).    
        \]

        When $i = 3$ we take care of each case separately depending on the generalized degree of $u_1$ and the structure of its neighbors. Let gd$(u_1) \leq $ gd$(u_2) \leq$ gd$(u_3)$. If gd$(u_2) = 2$ then let $x_1$ and $x_2$ be the two generalized neighbors.

        \begin{enumerate}
            \item gd$(u_1) = 0$. Every maximal forest $X \in \cM_G(F)$ will contain $u_1$, so $f(G,F,k) = f(G, F \cup \{u_1\}, k)$. This doesn't induce a branching vector.

            \item gd$(u_1) = 1$. Let $x$ be the generalized neighbor of $u_1$. Then every forest $X$ from $\cM_G(F)$ satisfies one fo the following conditions:
            \begin{itemize}
                \item either $X$ contains $v$;
                \item or $X$ does not contain $v$ but contains $u_1$;
                \item or $X$ does not contain $v$ and $u_1$ but contains $x$.
            \end{itemize}
            This means the number of maximal induced forests is at most
            \[
                f(G, F \cup \{v\}, k) + f(G \bs \{v\}, F \cup \{u_1\}, k-1) + f(G \bs \{v, u_1\}, F \cup \{x\}, k - 2).
            \]
            Since $x$ has weight at least $\wA$, then this will induce the following branching vector
            \[
                (\wA + 3(\wBeta - \wB), \wA + \wBeta + \wK, 2\wA + \wBeta + 2 \wK).
            \]

            
            \item gd$(u_1) = 2$, $x_1 \in N(t)$ and $x_1$ is also a generalized neighbor of either $u_2$ or $u_3$ (or both).
            We then have that every forest $X$ from $\cM_G(F)$ satisfies one of the following conditions:
            \begin{itemize}
                \item either $X$ contains $v$ and $x_1$ but does not contain $u_1$ and either $u_2$ or $u_3$ (or both);
                \item or $X$ contains $v$ and does not contain $x_1$;
                \item or $X$ does not contain $v$.
            \end{itemize}

            This means the number of maximal induced forests is at most
            \[
                f(G \bs \{u_1, u_2\}, F \cup \{v, x_1\}, k-2) + f(G \bs \{x_1\}, F \cup \{v\}, k-1) + f(G \bs \{v\}, F, k - 1).
            \]

            Since $x_1$ is a neighbor of $t$, and was chosen instead of $v$, then this means that $t$ has weight $\wA$, and has 3 generalized neighbors, two of which are $u_1$ and either $u_2$ or $u_3$. If at least three generalized neighbors of $x_1$ is $u_1$, $u_2$ and $u_3$ we have a more desirable branching, hence we assume that at most 2 generalized neighbors of $v$ are also generalized neighbors of $x_1$. But since at least 1 generalized neighbor of $x_1$ is not a generalized neighbor of $v$, then we gain at least a $\wBeta - \wB$ reduction when $x_1$ is chosen into the forest $X$.
            
            This case results in the branching vector
            \[
                (\wA + (\wA + (\wBeta - \wB)) + 2(\wBeta + \wK) ,\wA + \wA + 3(\wBeta - \wB) + \wK, \wA + \wK )
            \]


            \item gd$(u_1) = 2$, and $x_1 \in N(t)$ and $x_2 \in N(t)$. We note that due to previous cases, $x_1$ and $x_2$ are only generalized neighbors of $u_1$ and not $u_2$ or $u_3$. Then every forest $X$ from $\cM_G(F)$ satisfies one of the following conditions:
            \begin{itemize}
                \item either $X$ contains $v$ and $u_1$ but not $x_1$ and $x_2$;
                \item or $X$ contains $v$ but not $u_1$;
                \item or $X$ does not contain $v$.
            \end{itemize}

            This means the number of maximal induced forests is at most
            \[
                f(G \bs \{x_1, x_2\}, F \cup \{v, u_1\}, k-2) + f(G \bs \{u_1\}, F \cup \{v\}, k-1) + f(G \bs \{v\}, F, k - 1).
            \]
            In the first case, selecting $v$ of weight $\wA$ also reduces the measure by $2(\wBeta - \wB)$, one for each of $u_2$ and $u_3$. Selection of $u_1$ reduces the measure by $\wB$ and removing $x_1$ and $x_2$ results in a reduction of at least $2(\wA + \wK)$. The second case also selects $v$, but removes $u_1$ for a measure decrease of $\wBeta + \wK$. The third case just removes $v$, for a total decrease of $\wA + \wK$.
            which results in the branching vector
            \[
                (\wA + 2(\wBeta - \wB) + \wBeta + 2 (\wA + \wK), \wA + 2(\wBeta - \wB) + \wBeta + \wK , \wA + \wK).
            \]
            
            \item gd$(u_1) = 2$, and previous subcases don't apply. At least one of $x_1$ and $x_2$ has weight $\wBeta$ otherwise we would be in case (b). Let $x_1$ has weight $\wBeta$. Then every forest $X$ from $\cM_G(F)$ satisfies one of the following conditions:
            \begin{itemize}
                \item either $X$ contains $v$;
                \item or $X$ does not contain $v$ but contains $u_1$;
                \item or $X$ does not contain $v$ and $u_1$ but contains $u_2$ and $x_1$;
                \item or $X$ does not contain $v$, $u_1$ and $x_1$ but contains $u_2$ and $x_2$;
                \item or $X$ does not contain $v$, $u_1$ and $u_2$ but contains $u_3$ and $x_1$;
                \item or $X$ does not contain $v$, $u_1$, $u_2$ and $x_1$ but contains $u_3$ and $x_2$.
            \end{itemize}
            Now we consider the weight of $x_2$, which is only ever selected into the forest. In these cases, if $x_2$ has weight $\wBeta$ we simply reduce the measure by $\wBeta$. Now $x_2$ cannot has weight $\wA$ since case 6 did not occur. So if $x_2$ is of weight $\wA$, it has at least 3 generalized neighbors of which only 1 is $u_1$ and $u_2$ and $u_3$, are not generalized neighbors of $x_2$. This means that when $x_1$ is selected into the forest $X$ we have another 2 $(\wBeta - \wB)$ reductions in the worst case.

            To simplify our analysis into a single branching vector, we enforce that the $\wBeta$ weight reduction is the worst case
            \[
                \wBeta \leq \wA + 2(\wBeta - \wB).
            \]
            This effectively means that if vertex $x_2$ is ever selected, then in the worst case there is a reduction of measure of at least $\wBeta$.
            We thus obtain the branching vector
            \begin{align*}
                (\wA + 3(\wBeta - \wB), \wA + \wBeta + \wK, \wA + 3\wBeta + 2\wK, \\
                \wA + 4\wBeta + 3\wK, \wA + 4\wBeta + 3\wK, \wA + 5\wBeta + 4\wK).
            \end{align*}
            
            \item gd$(u_1) \geq 3$. Every forest $X$ from $\cM_G(F)$ satisfies one of the following conditions
            \begin{itemize}
                \item either $X$ contains $v$
                \item or $X$ does not contain $v$.
            \end{itemize}

           This means the number of maximal induced forests is at most
            \[
                f(G, F \cup \{v\}, k) + f(G \bs \{v\}, F, k - 1).
            \]

            Since all of $u_i$ for $i = 1,2,3$ have at least 3 generalized neighbors, then we obtain the branching vector
            \[
                (\wA + 3(\wBeta - \wA), \wA + \wK).
            \]
        \end{enumerate}

        \item gd\emph{$(v) \geq 4$.} Every forest $X$ from $\cM_G(F)$ either contains $v$ or doesn't contain $v$. Hence the number of forests is upper bounded by
        \[
            f(G, F \cup \{v\}, k) + f(G \bs \{v\}, F, k - 1)
        \]
        which results in the branching vector
        \[
            (\wA + 4(\wBeta - \wB), \wA + \wK).
        \]
    \end{enumerate}


    \subsection{Results}

    \begin{restatable}{theorem}{thmMFVSresult}
        For a graph $G$ with $n$ vertices, all minimal feedback vertex sets can be enumerated in time \rtenumfvs.
    \end{restatable}

    \begin{proof}
        We evaluate the running time of the proposed algorithm above using the stated measure $\mu$. It can be shown that the weights
        \[
            \wA = 0.4506859777, \quad
            \wB = 0.4233244855, \quad
            \wBeta = 0.7809613776, \quad
            \wK = 0.2081356098, \quad
        \]
        satisfy all stated branching factors and constraints necessary.
        The number of maximal induced forests containing $F$ of size at least $n - k$ is upper bounded by 
        \[
            f(\mu, k) \leq 2^{\wBeta} \cdot 2^{\wK} \leq 1.7183^n 1.1552^k
        \]
        This results in a $O(1.7183^n 1.1552^k)$ algorithm for enumerating the maximal induced forests of size at least $n-k$, and also enumerating minimal feedback vertex sets of size $k$. 
        Consider now the enumeration of the collection
        \[
            \cF_{I,X}^k = \{S \subseteq V \backslash X : |S| = k \text{ and $S \cup X$ is a minimal FVS} \}.
        \]
        By running the new algorithm just described on the subgraph $G[V \bs X]$ that remains after removing the vertices of $X$, we enumerate all minimal feedback vertex sets of size $k$ in time $1.7183^n \cdot 1.1552^k \cdot N^{O(1)}$. For every minimal feedback vertex set $S$ that was just enumerated, we can check in polynomial time if $S \cup X$ is also a minimal feedback vertex set. This means that the collection $\cF_{I,X}^k$ can thus also be enumerated in time $1.7183^{n - |X|}  \cdot 1.1552^k \cdot N^{O(1)}$.
        Combined with Theorem \ref{thm:main4} this results in a \rtenumfvs deterministic algorithm for the number of minimal feedback vertex sets in a graph $G$.
    \end{proof}

\section{Minimal Hitting Sets}
    Based on \cite{cochefert2016faster} we once again apply a multivariate analysis to enumerating all minimal hitting sets on a hypergraph of rank 3.

    \subsection{Measure}

    We conduct a new analysis using the algorithm described in \cite{cochefert2016faster}, by first deriving a similar measure. Let $H$ be a hypergraph of rank 3 and $k$ be an upper bound on the size of the hitting set $S \subseteq V$. Denote by $n_i$ the number of vertices of degree $i \in \mathbb{N}$ and $m_i$ the number of hyperedges of size $i \in \{0,...,3\}$. Also denote $m_{\leq i} := \sum_{j = 0}^i m_j$. Then a measure for $H$ and a given $k$ is
    \[
        \mu(H, k) = \varPsi(m_{\leq 2}) + \sum_{i = 0}^\infty w_i n_i + w_k \cdot k
    \]
    where $\varPsi : \mathbb{N} \to \mathbb{R}_{\geq 0}$ is a non-increasing non-negative function independent of $n$, and $\omega_i$ are non-negative reals. Clearly $\mu(H, k) \geq 0$. 

    We make the same simplifying assumptions as \cite{cochefert2016faster} which provides the constraints
    \begin{align}
        \omega_i &:= \omega_5, &\varPsi(i) &:= 0 & \text{ for each } i \geq 6 \\
        \Delta \omega_i &:= \omega_i - \omega_{i-1},  &\Delta \varPsi(i) &:= \varPsi(i) - \varPsi(i-1) & \text{ for each } i \geq 1 \\
        0 &\leq \delta \omega_{i+1} \leq \Delta \omega_i, &\text{ and } 0 &\geq \Delta \varPsi(i+1) \geq \Delta \varPsi(i) & \text{ for each } i \geq 1
    \end{align}

    \noindent
    Further branching rules will add constraints on the measure. Denote $T(\mu) := 2^{\mu}$ as an upper bound on the number of leaves of the search tree modelling the recursive algorithm for all $H$ with $\mu(H) \leq \mu$.

    \subsection{Analysis}

    \begin{restatable}{theorem}{thmMHSresult}
        For a hypergraph $H$ with $n$ vertices and rank 3, all minimal hitting sets can be enumerated in time \rthittingset.
    \end{restatable}

    \begin{proof}
        We follow the algorithm and rules from \cite{cochefert2016faster}, but adding an additional weight $w_k$ to the decrease in measure every time a vertex is to be selected into $S$, the partial hitting set for the hypergraph $H$. This is applied across all cases and constraints as outlined in the algorithm. All constraints are satisfied with the given weights for the measure $\mu$.

        \begin{figure}[ht]
            \centering
            \begin{tabular}{l l l}
                $i$ & $\omega_i$ & $\varPsi(i)$ \\
                \hline
                0 & 0            & 0.597858842  \\
                1 & 0.506736167  & 0.396229959  \\
                2 & 0.582087585  & 0.238416719  \\
                3 & 0.595274776  & 0.094238670  \\
                4 & 0.597858842  & 0.030265172  \\
                5 & 0.597858842  & 0.004148530  \\
                6 & 0.597858842  & 0            \\
            \end{tabular}
            \begin{tabular}{c}
                $w_k = 0.2330535427$
            \end{tabular}
        \end{figure}
        \noindent
        The number of leaves in the search tree is upper bounded $T(\mu) \leq 2^{w_6 \cdot n + w_k \cdot k}$. These weights result in a multivariate running time of $O(1.5135^n \cdot 1.1754^k)$ for enumerating minimal hitting sets of size at most $k$ in rank 3 hypergraphs. Then the collection
        \[
            \cF_{I,X}^k = \{S \subseteq V \backslash X : |S| = k \text{ and $S \cup X$ is a minimal 3-HS} \}
        \]
        can be enumerated in time $1.5135^{n - |X|} \cdot 1.1754^k \cdot N^{O(1)}$.
        Combined with Theorem \ref{thm:main4} this results in an algorithm for enumerating minimal hitting sets in rank 3 hypergraphs in \rthittingset.

    \end{proof}

\section{Conclusion}
    The main contribution of this paper is a framework allowing us to turn many $b^n c^k \bitsize^{O(1)}$ time algorithms for subset and subset enumeration problems into $(1+b-\frac{1}{c})^n N^{O(1)}$ time algorithms, generalizing a recent framework of Fomin et al. \cite{FominGLS16}.

    The main complications in using the framework are, firstly, that new algorithms or running-time analyses are often needed, and, secondly, that such analyses need solutions to non-convex programs in the Measure and Conquer framework. In the usual Measure and Conquer analyses \cite{FominGK09},the objective is to upper bound a single variable ($\alpha$) which upper bounds the exponential part of the running time ($2^{\alpha n}$) subject to convex constraints. Thus, it is sufficient to solve a convex optimization problem to minimize the running time \cite{gaspers2010exponential,GaspersS12} resulting from the constraints derived from the analysis. Here, the objective function ($2^{\alpha}-2^{-w_k}$) is non-convex. While experimenting with a range of solvers, either guaranteeing to find a global optimum (slow and used for optimality checks) or only a local optimum (faster and used mainly in the course of the algorithm design), we experienced on one hand that the local optima found by solvers are often the global optimum, but on the other hand that weakening non-tight constraints can sometimes lead to a better globally optimum solution.

\paragraph*{Acknowledgments}
    We thank Daniel Lokshtanov, Fedor V. Fomin, and Saket Saurabh for discussions inspiring some of this work.
    Serge Gaspers is the recipient of an Australian Research Council (ARC) Future Fellowship (FT140100048) and acknowledges support under the ARC's Discovery Projects funding scheme (DP150101134).

\bibliographystyle{plain}
\bibliography{main}


\end{document}